\documentclass[conference, letterpaper]{IEEEtran}
\usepackage[top=0.7in, bottom=1.0in, left=0.65in, right=0.65in]{geometry}

\usepackage[short]{optidef}

\usepackage{pgfplots}
\usepackage{tikz}
\usetikzlibrary{calc}
 
\newcommand{\gettikzxy}[3]{%
  \tikz@scan@one@point\pgfutil@firstofone#1\relax
  \edef#2{\the\pgf@x}%
  \edef#3{\the\pgf@y}%
}

\usepackage{setspace}
\usepackage{lettrine}

\usepackage{graphics,
	psfrag,
	epsfig,
	amsthm,
	cite,
	amssymb,
	url,
	dsfont,
	subfigure,
	balance,
	enumerate,
	color,
	setspace
}

\usepackage{multirow}
\usepackage{algorithm} 
\usepackage{algorithmicx}
\usepackage{algpseudocode}

\usepackage{amsmath}

\usepackage{epstopdf}

\newcommand{\ignore}[1]{}

\newtheorem{proposition}{Proposition}

\usepackage{comment}
\usepackage{derivative}

\newcommand{\eref}[1]{(\ref{#1})}

\newcommand{\cref}[1]{Constraint~\ref{#1}}

\IEEEoverridecommandlockouts
\usepackage{smartdiagram}

\usetikzlibrary{spy,backgrounds}
\usetikzlibrary{shapes, arrows, positioning}

\usepackage{mathtools}
\usepackage{color}
\usepackage[utf8]{inputenc}

\begin{document}

\title{
 Maximizing Connectivity of Uplink RIS-Assisted UAV Networks 
}

\author{\IEEEauthorblockN{Mohammed Saif\textsuperscript{1} and Shahrokh Valaee\textsuperscript{2}}
\IEEEauthorblockA{\textsuperscript{1}Electrical, Computer, and Biomedical Engineering, Toronto Metropolitan University, Toronto, ON, Canada \\
\textsuperscript{2}Electrical and Computer Engineering, University of Toronto, Toronto, ON, Canada\\
Email: mohammed.saif@torontomu.ca, valaee@ece.utoronto.ca}
}

\IEEEoverridecommandlockouts
\maketitle

\IEEEpubidadjcol

\begin{abstract}
In this paper, we present a new approach for unmanned aerial vehicle (UAV) positioning and reconfigurable intelligent surface (RIS) partitioning to enhance connectivity of uplink RIS-assisted UAV networks. To achieve this, our approach optimizes RIS-aided link selection, RIS partitioning, and UAV positions to maximize network connectivity characterized by its Fiedler value. Meanwhile, it maintains a specific signal-to-interference plus noise ratio (SINR) constraint for user equipment (UE), which is influenced by RIS partitioning and UAV reliability. The network connectivity optimization problem is formulated using the Fiedler value subject to RIS elements allocation and SINR constraints. This problem is a computationally expensive combinatorial optimization, necessitating an efficient iterative approach. In particular, we propose a perturbation method for RIS-aided link selection, and  derive a closed-form solution for RIS partitioning, with each partition tailored to optimize SINR for individual UAV.  For the given RIS-aided links and RIS partitioning, we then show that the problem of UAV positioning can be formulated as a low complexity semi-definite programming (SDP) optimization problem, which can be solved using off-the-shelf CVX solvers. Our simulations show the potential gain of UAV positioning and RIS partitioning compared to the benchmark schemes from the literature.

\end{abstract}


\section{Introduction}
Optimizing the locations of unmanned aerial vehicles (UAVs) is essential for a wide range of applications, including extending network coverage, enhancing connectivity and resiliency, and improving localization and tracking \cite{9716042, 8292633}. Additionally, the flexible deployment of UAVs allows the establishment of line-of-sight (LoS) communications to geographically distant user equipments (UEs), enabling a large number of UEs to connect to the network and reducing the probability of unconnected UEs. Since UAVs are prone to failure due to limited energy, deploying additional UAVs is not always preferable. 

Reconfigurable intelligent surfaces (RISs) have emerged as a low-cost solution for controlling the wireless
environment, with significant potential for improving network connectivity and extending service coverage. By manipulating the propagation environment, RISs can significantly enhance energy efficiency \cite{10437618}, coverage \cite{9756313}, and network connectivity \cite{10104574}. From a connectivity perspective, RISs can introduce additional links to networks and significantly enhance connectivity.
RISs can also mitigate UAV failures—measured by the removal of UAV nodes from the network—by redirecting UE signals to reliable UAVs, thereby maintaining network integrity.  

Most studies on maximizing network connectivity have primarily focused on utilizing UAVs \cite{8292633}, relays \cite{4786516}, or sensors \cite{4657335} in different systems.  The authors of \cite{8292633} maximize network connectivity, modeled by the Fiedler value (i.e., the second smallest eigenvalue of the Laplacian graph network \cite{new}), by optimizing UAV positioning in small-cell systems.  The works \cite{4786516, 4657335} consider optimizing the placement of relays and sensors for connectivity maximization and network repair maintenance. On the other hand, the work \cite{10104574} improves network connectivity and resiliency of cell-free networks using RIS deployment. 
The convergence of RIS and UAV technologies has been leveraged for enhancing key performance metrics, including physical layer security \cite{10507188} and connectivity \cite{saifglobecom_E}. Integrating RIS technology with UAVs introduces numerous optimization challenges, such as UAV positioning, UE-RIS-UAV link selection, and RIS phase design \cite{9756313}. 

The work  \cite{saifglobecom_E} improves connectivity of RIS-assisted UAV networks using RIS deployment, with the optimization of  UE-RIS-UAV link selection and RIS phase design only. Specifically, \cite{saifglobecom_E} proposes to generate one UE-RIS-UAV link from each RIS, ignoring the efficient use of RIS that could be utilized to generate multiple links in the network. In addition, they do not consider optimizing UAV location to further improve network connectivity and coverage. Recent advancements in RIS optimization have introduced a technique known as RIS virtual partitioning, which significantly enhances communication performance \cite{10458024}. RIS virtual partitioning divides a single RIS into multiple virtual sections, each configured with distinct phase shifts optimized for specific directions. This innovative approach enables the RIS to reflect an incoming signal into multiple cascaded signals, each directed toward different targets, thereby enhancing spatial efficiency and signal coverage \cite{10458024}. Unlike \cite{10507188}, which investigates RIS virtual partitioning and UAV communications for physical-layer security, this paper focuses on leveraging these techniques to enhance network connectivity. In particular, the RIS is utilized to establish multiple communication links—each originating from a virtual section—thereby improving connectivity.


The main findings of this
paper can be summarized as follows. 
\begin{itemize}
    \item The network connectivity optimization problem is formulated using the Fiedler value subject to RIS elements allocation and signal-to-interference plus noise ratio (SINR) constraints. 

    \item To tackle this problem, we develop an iterative optimization framework. For given UAV positions and RIS partitioning, we first propose a connectivity-aware perturbation method to optimize the RIS-assisted link selection. We derive a closed-form solution for RIS partitioning. Then, for the resulting link selection and RIS partitioning, we reformulate the UAV positioning subproblem as a low-complexity semidefinite programming (SDP) problem, which can be efficiently solved using standard CVX solvers. The proposed procedure iteratively updates these variables until convergence.

    \item Our simulations show the potential gain of the proposed iterative solution compared to the benchmark schemes from the literature.
\end{itemize}


\section{System Model}\label{S}
 
Consider an uplink-RIS-assisted UAV model, where a set of single-antenna UEs, denoted by $\mathcal M=\{1, 2, \ldots, M\}$, transmit signals to the base station (BS) via a set of single-antenna UAVs, denoted by $\mathcal K=\{1, 2, \ldots, K\}$, with the aid of RIS comprising of $N$ passive elements. 
Through UAV positioning optimization and RIS deployment, service coverage area can be extended to cover the $M$ UEs, which are in deep fade situation. 





The RIS can reflect the signal of the transmitted UE to multiple UAVs through virtual partitioning representation; however, the RIS cannot be assigned to more than one UE at the same time \cite{10507188}.  The set of multiple UAVs that exploit the RIS  is denoted as $\mathcal K_\text{RIS}$, which also represents the number of partitions in the RIS, i.e., $\mathcal K_\text{RIS}=\{\text{UAV}_1, \ldots, \text{UAV}_{K_\text{RIS}}\}$.  We represent the RIS elements’ allocation portions $\boldsymbol{\alpha}=[\alpha_1, \ldots,  \alpha_{K_\text{RIS}}]$ to denote the RIS portions that are allocated to UAV$_{k}$, $k \in \{1, \ldots, K_\text{RIS}\}$, respectively. The number of RIS elements designated for UAV$_k$ is $N_k$, which  can be presented as $N_k=\lceil\alpha_kN\rceil$. UAV$_1$ obtains a coherently aligned link from the designated RIS portion $\alpha_1$ and a non-coherently aligned links from the remaining RIS portions $\{\alpha_2, \ldots, \alpha_{K_\text{RIS}} \}$.


 

All channels of UE-UAV, UE-RIS, and RIS-UAV   are considered to be quasi-static with flat-fading and follow the Nakagami-$f$ fading model \cite{10507188}. Let  $(\mathbf g_m^\text{MR}, \beta_m^\text{MR})$ and $(\mathbf g^\text{RK}_k ,\beta^\text{RK}_{k})$ represent the small-scale fading coefficients and path-losses for the
$\text{UE}_m \rightarrow \text{RIS}$ and $\text{RIS} \rightarrow \text{UAV}_{k}$ channels, respectively. Here, $\mathbf g_m^\text{MR}= [g^\text{MR}_{1,m}, \ldots, g^\text{MR}_{N,m}]$ and $\mathbf g^\text{RK}_{k}= [g^\text{RK}_{1,k}, \ldots, g^\text{RK}_{N,k}]$. Therefore, we consider the channel from UE$_m$ to UAV$_k$ over the $n$-th RIS element as $g^\text{MRK}_{m,n,k}= g^\text{MR}_{n,m} g^\text{RK}_{n,k}$, where $g^\text{MR}_{n,m}=|g^\text{MR}_{n,m}|e^{-j\phi_{n,m}}$ denotes the channel coefficient between UE$_m$ and the $n$-th RIS element, while $g^\text{RK}_{n,k}=|g^\text{RK}_{n,k}|e^{-j\psi_{n,k}}$ is the channel coefficient
between the $n$-th RIS element and UAV$_k$; $|g^\text{MR}_{n,m}|$ and $|g^\text{RK}_{n,k}|$ are the channel amplitudes, while  $\phi_{n,m}$ and  $\psi_{n,k}$ are the channel phases. Moreover, for the direct links between the UEs and the UAVs, let $g^\text{MK}_{m,k}$ and $\beta^\text{MK}_{m,k}$ denote the small-scale fading coefficient and path-loss for the $\text{UE}_m \rightarrow \text{UAV}_k$ channel, respectively.

\begin{table*}
	\begin{align}\label{exact} 
 \text{SINR}_{k}(\boldsymbol{\alpha})&= \frac{p \bigg| \underbrace{\sqrt{\beta_{m,k}^\text{UK}}g^\text{MK}_{m,k}}_{\textbf{direct link}}+  \underbrace{\sqrt{\beta_m^\text{MR}\beta^\text{RK}_{k}}\sum_{n=1}^{N_k}g^\text{MRK}_{m,n,k}e^{j\theta_{n}}}_\textbf{aligned phase}+\underbrace{\sum_{k'=1, k'\neq k}^K\sqrt{\beta_m^\text{MR}\beta^\text{RK}_{m,k'}}\sum_{n'=1}^{N_{k'}}g^\text{MRK}_{m,n,k'}e^{j\theta_{n'}}}_\textbf{non-aligned phase}\bigg|^2}{\underbrace{\sum_{m'=1, m' \neq m}^M \sqrt{\beta_{m',k}^\text{MK}}g^\text{MK}_{m',k}}_{\textbf{direct links of other UEs}}+N_0}  \\ &  \nonumber \overset{(a)}=
\frac{p \bigg| \sqrt{\beta_{m,k}^\text{MK}}g^\text{MK}_{m,k}+  \sqrt{\beta_m^\text{MR}\beta^\text{RK}_{k}}\alpha_k \sum_{n=1}^{N}|g^\text{MR}_n||g^\text{RK}_{n,k}| \bigg|^2}{\sum_{m'=1, m' \neq m}^M \sqrt{\beta_{m',k}^\text{MK}}g^\text{MK}_{m',k}+N_0} \overset{(b)}=    
\frac{p \bigg| \sqrt{\beta_{m,k}^\text{MK}}g^\text{MK}_{m,k}+  \sqrt{\beta_m^\text{MR}\beta^\text{RK}_{k}}\alpha_k Q\bigg|^2}{\sum_{m'=1, m' \neq m}^M \sqrt{\beta_{m',k}^\text{MK}}g^\text{MK}_{m',k}+N_0}.  
\end{align}
\hrulefill
	\vspace*{-0.5cm}
\end{table*}

The SINR at UAV$_k$ can be expressed as in \eref{exact} given at the top of the next page, where $N_0$ is the additive white Gaussian noise (AWGN) power density, $\theta_n$ is the phase shift of element $n$, $p$ is the transmit power of an UE, $\overset{(a)}=$ comes from (i) $\sum_{n=1}^{N_k}(.)=\sum_{n=1}^{\lceil\alpha_kN\rceil}(.)= \alpha_k \sum_{n=1}^{N}(.)$ and (ii) $g^\text{MRK}_{n,k}e^{j\theta_{n}}=|g^\text{MR}_n||g^\text{RK}_{n,k}|e^{j(\theta_{n}-\phi_n-\psi_{n,k})}=|g^\text{MR}_n||g^\text{RK}_{n,k}|$; i.e., for the $n$-th RIS element, given $\phi_n$ and $\psi_{n,k}$, the RIS's controller can perfectly align $\theta_n$ with $\phi_n$ and $\psi_{n,j}$ to nullify their effect, and $\overset{(b)}=$ comes from $Q=\sum_{n=1}^N|g^\text{MR}_n||g^\text{RK}_{n,k}|$ which is the $N$-element double-Nakagami$-f$ that is independent and
identically distributed (i.i.d.) random variable (RV) with parameters $f_1$, $f_2$, $\Omega_1$, and $\Omega_2$; the distribution of the product
of two RVs following the Nakagami$-f$ distribution with the probability density function (PDF) is given in \cite{10507188}. To ease the analysis of the RIS partitioning optimization, the non-aligned phase term is ignored in \eref{exact}. Nonetheless, our numerical results show that the impact of this term is minimal.

\section{Problem Modeling}\label{PF}

\subsection{Graph Construction and Node Reliability} 
We model the uplink RIS-assisted UAV network using an undirected, weighted, simple finite graph $\mathcal G(\mathcal V, \mathcal E)$, where $\mathcal V=\{v_1, \ldots, v_{V} \}$ is the set of vertices (UE and UAV nodes) and $\mathcal E=\{e_1, \ldots, e_E \}$ is the set of $E$ edges connecting the nodes, i.e., links connecting  the UE to UAVs and links among the UAVs. For any
edge $l$ between two vertices $v_i$ and $v_j \in \mathcal V$, the edge
vector $\mathbf a_l \in \mathbf R^{V}$ is  a vector of zeros except for its $i$-th and
$j$-th elements, where $a_{l,i}=1$  and $a_{l,j}=-1$, respectively. The graph's incidence matrix of $\mathbf A \in \mathbf R^{V\times E}$ is constructed as $\mathbf A = [\mathbf a_l, \ldots, \mathbf a_E]$. The Laplacian matrix of the graph, $\mathbf L \in R^{V\times V}$, is then given as follows \cite{8292633}
\begin{equation} \label{lap}
\mathbf L= \mathbf A  ~diag(\mathbf w) ~\mathbf A^T=\sum^{E}_{l=1} w_l \mathbf a_l \mathbf a^T_l,
\end{equation}
where $\mathbf w \in {[\mathbb R^+]}^E$ denotes the $E \times 1$ weighting vector coefficients
for the $E$ edges and is given by   $\mathbf w= [ w_1, w_2, \ldots, w_E]$. The weights of the edges are based on their corresponding SNRs that are modeled as in \cite{8292633}, while the connections of the edges are based on SNR thresholds, i.e.,  $\gamma^\text{UE}_{0}$ and $\gamma^\text{UAV}_{0}$ for $\text{UE} \rightarrow  \text{UAV}_{k}$ and $\text{UAV}_k \rightarrow  \text{UAV}_{k'}$, respectively. The Laplacian matrix $\mathbf L$ is a positive semi-definite matrix, i.e., $\mathbf L \geq 1$. Its second smallest eigenvalue, denoted by $\lambda_2 (\mathbf L)$, is known as the Fiedler value. A higher Fiedler value indicates stronger overall network connectivity, while $\lambda_2 (\mathbf L)=0$ implies that the graph is disconnected \cite{new}.

The deployment of UAV positioning and RIS partitioning  can create more links, thereby a new graph $\mathcal G'(\mathcal V, \mathcal E')$ is built with the same number of $V$ nodes and a larger set of edges denoted by $\mathcal E'$ with $\mathcal E'=\mathcal E \cup \mathcal E_{new}$, where $\mathcal E_{new}$ is the new edges for the  $\text{UE} \rightarrow  \text{UAV}_{k\in \mathcal K_\text{RIS}}$  links. In particular, such deployment has an impact on relaying information from UEs to UAVs and extends the coverage, creating additional $E'-E$ edges to the original network. By comparing the network graphs before and after the deployment, the connectivity improvement can be quantified by evaluating the change in the Fiedler value. Specifically, the gain is observed when
$\lambda_2 (\mathbf L') \geq \lambda_2 (\mathbf L)$, where $\mathbf L'$ represents the Laplacian matrix of the updated graph $\mathcal G'(\mathcal V, \mathcal E')$.

Each UAV has a different impact on network connectivity when it is removed from the graph along with its connected edges to other nodes. Similar to \cite{saifglobecom_E}, we measure the reliability of UAV$_{k \in \mathcal{K}}$ based on their impact on network connectivity, defined as $\mathcal{R}_k = \lambda_2(\mathcal{G}_{-k})$, where $\mathcal{G}_{-k}$ is the subgraph resulting from removing UAV$_k$ and all its adjacent edges to other nodes in $\mathcal{G}$.

\subsection{Problem Formulation}

Let $\mathbf X$ be the binary UE-RIS  assignment matrix with entries $x_{m}$, where  $x_{m}$ is $1$ if UE$_m$ is connected to the RIS, and  $x_{u}=0$ otherwise. Likewise, let $\mathbf Z$ be the  RIS-UAV assignment matrix and comprise binary elements $z_k$  defined as $1$ if  $\text{UAV}_k$ is connected to the RIS, and $0$ otherwise. Let $\gamma_\text{th}$ be the SINR threshold. Mathematically, the problem can be written as follows
\begin{align}\nonumber
\label{P0}
&\max_{\boldsymbol{0} \preceq \boldsymbol{\alpha} \preceq \boldsymbol{1},\boldsymbol{u}, \mathbf X, \mathbf Z} ~~~~ \lambda_2(\mathbf L' (\boldsymbol{\alpha}, \boldsymbol{u}, \mathbf X, \mathbf Z))
 \\ 
 &{\rm s.~t.\ } ~~~~~\sum_{m=1}^Mx_m\leq 1, \\ \nonumber &
~~~~~~~~~~~ 1 \leq \sum_{k =1}^K z_k \leq K_\text{RIS}, \\ \nonumber &
~~~~~~~~~~~\text{SINR}_k \geq \mathcal R_k\gamma_\text{th},~~ \forall k \in \{1, \ldots,  K_\text{RIS}\},\\ \nonumber 
& ~~~~~~~~~~~\sum_{k=1}^{K_\text{RIS}}\alpha_k\leq 1,  \\ \nonumber 
& 
~~~~~~~~~~~x_m, z_k \in \{0,1 \}, \nonumber 
\end{align} 
where $\boldsymbol{u}=[\boldsymbol{u}_1, \ldots, \boldsymbol{u}_{K_\text{RIS}}]$, where $\boldsymbol{u}_1$  is the $3 \times 1$ position vector in a Cartesian
3D coordinate system of UAV$_1$, and $\preceq$ is the pairwise inequality. The first constraint ensures a single UE is connected to the RIS, while the second constraint implies that at most one reflected link should be created for UAV$_k$  via the RIS. The third constraint defines the SINR requirement for UAV$_{k}$, while the second constraint ensures that the allocated portions do not exceed one, so that the total number of assigned RIS elements does not surpass the available RIS elements. The problem in \eref{P0} is a computationally expensive combinatorial optimization.

\section{Proposed Solution}\label{PS}
This section tackles \eref{P0} by decomposing it into three subproblems and solving them  iteratively.  

\subsection{UE-RIS-UAV Link Selection}
We propose a low-complexity yet effective heuristic algorithm for jointly optimizing $\mathbf{X}$ and $\mathbf{Z}$ based on the Fiedler vector of the original network graph. The key idea is to iteratively establish the $K_{\mathrm{RIS}}$ RIS-assisted links that provide the largest improvement in network connectivity. Instead of repeatedly recomputing the algebraic connectivity, the proposed method exploits the entries of the Fiedler vector and evaluates the weighted differences between connected node pairs derived from the graph Laplacian $\mathbf{L}$. This significantly reduces the computational complexity while effectively identifying the most beneficial RIS-assisted links. In the following proposition, we show the upper bound for $\lambda_2(\mathbf L')$, and describe the proposed perturbation heuristic.

\begin{proposition}\label{prep1}
$\lambda_2(\mathbf L')$ is upper bounded by $ \lambda_2(\mathbf L)+ \mathcal R_k(v_m-v_k)^2$, where  $v_m$ and $v_k$ are the corresponding values of UE$_m$ and UAV$_k$ indices of the Fiedler vector $\mathbf v$ of $\lambda_2 (\mathbf L)$. 
\end{proposition}

Let $W_{mk}=(v_m-v_k)^2$. 
Based on Proposition~\ref{prep1}, the proposed perturbation heuristic aims to maximize the increase in the algebraic connectivity by selecting RIS-assisted links with the largest weighted Fiedler distance. Starting from the original graph \(\mathcal{G}\) with Laplacian matrix \(\mathbf{L}\), the proposed algorithm proceeds as follows:

\begin{itemize}
\item Compute the Fiedler vector \(\mathbf{v}\), i.e., the unit eigenvector associated with the second-smallest eigenvalue \(\lambda_2(\mathbf{L})\). For each feasible UE-RIS-UAV schedule, evaluate the gain metric $\mathcal{R}_k W_{mk}$.

\item Among all remaining feasible schedules, select the UE-RIS-UAV link that yields the largest value of \(\mathcal{R}_k W_{mk}\), and add the corresponding edge to the graph.

\item Update the graph topology and repeat the above procedure until \(K_{\mathrm{RIS}}\) RIS-assisted links have been established for the selected UE or no additional feasible links remain.

\end{itemize}


\subsection{RIS Partitioning Optimization}
For given  $\boldsymbol{u}$, $\mathbf X, \mathbf Z$, we  optimize $\boldsymbol{\alpha}$  to maximize
the sum SINR of the newly established $\text{UE} \xrightarrow{\text{RIS}} \text{UAV}_{k\in \mathcal K_\text{RIS}}$ links,   considering the
SINR constraint based on UAV reliability and RIS elements. Thus,  the subproblem for updating
RIS partitioning  is given by  
\begin{align}\label{RISp}
&\max_{\boldsymbol{0} \leq \boldsymbol{\alpha} \leq \boldsymbol{1}} ~~~~  \lambda_2(\mathbf L' (\boldsymbol{\alpha}))  
\\
 &{\rm s.~t.\ } ~~~~~~~\text{SINR}_k(\boldsymbol{\alpha})\geq \mathcal R_k\gamma_\text{th},~~ \forall k \in \{1, \ldots,  K_\text{RIS}\}, \nonumber \\&  ~~~~~~~~~~~~\sum_{k=1}^{K_\text{RIS}}\alpha_{k}\leq 1.\nonumber 
\end{align}

The proceeding proposition
provides the closed-form solution $\boldsymbol{\alpha}^*$ of \eref{RISp}.  
\begin{proposition}
The RIS partitioning that provides the maximized network connectivity via maximizing the sum SINR of the new $\text{UE} \xrightarrow{\text{RIS}} \text{UAV}_{k\in \mathcal K_\text{RIS}}$ links is given by
\begin{equation}\label{closed}
\alpha^*_{k}= \begin{cases}
\sqrt{\frac{\mathcal R_k \gamma_\text{th} (\Tilde{\gamma}+N_0)-p\beta^\text{MK}}{p\beta^\text{MR}\beta^\text{RK}_k N^2z}}, & \text{if}  ~  k \neq k^*,\\
1-\sum_{k=1}^{K_\text{RIS}-1}\alpha^{*}_{k}, & \text{if}  ~  k=k^*.
\end{cases}
\end{equation}
 \end{proposition}
\begin{proof}
Due to the complexity of obtaining $\boldsymbol{\alpha}$ directly from \eref{exact} and the first constraint of \eref{RISp}, which involves expansion and derivative, we first approximate the SINR expression in \eref{exact} by utilizing the expected values of $\sqrt{\beta_{m,k}^\text{MK}}|g^\text{MK}_{m,k}|, \sqrt{\beta_{m',k}^\text{MK}}|g^\text{MK}_{m',k}|$, and $\alpha_{k} Q$. Specifically, the expectations can be evaluated as $\text {E} \bigg\{ \bigg|\sqrt{\beta_{m,k}^\text{MK}}|g^\text{MK}_{m,k}| \bigg|^2\bigg\}=\beta_{m,k}^\text{MK}$, since $g^\text{MK}_{m,k} \sim \mathcal {CN}(0,1)$,  $\text {E}\big\{|g^\text{MK}_{m,k}|^2\big\}=1$; $\text {E} \bigg\{ \bigg|  \sqrt{\beta_m^\text{MR}\beta^\text{RK}_k}\alpha_{k} Q\bigg|^2\bigg\}= \beta_m^\text{MR}\beta^\text{RK}_k \text {E} \bigg\{ \bigg|  \alpha_{k} \sum_{n=1}^N|g^\text{MR}_{n,m}||g^\text{RK}_{m,n,k}|\bigg|^2\bigg\}= \beta_m^\text{MR}\beta^\text{RK}_k \alpha^{2}_{k} N^2 \frac{1}{\Tilde{f} \hat{f}} \frac{\Gamma(\Tilde{f}+0.5)^2}{\Gamma(\Tilde{f})^2} \frac{\Gamma(\hat{f}+0.5)^2}{\Gamma(\hat{f})^2}=\beta_m^\text{MR}\beta^\text{RK}_k \alpha^{2}_{k} N^2z$, where $z=\frac{1}{\Tilde{f} \hat{f}} \frac{\Gamma(\Tilde{f}+0.5)^2}{\Gamma(\Tilde{f})^2} \frac{\Gamma(\hat{f}+0.5)^2}{\Gamma(\hat{f})^2}$, where $\Gamma(.)$ denotes the Gamma function. By substituting these values in \eref{exact}, we have
\begin{align}\label{appro1}
  &\text{SINR}_k(\boldsymbol{\alpha})=  \frac{p[\beta_{m,k}^\text{MK}+\beta_m^\text{MR}\beta^\text{RK}_k \alpha^{2}_{k} N^2z]}{\sum_{m'=1, m'\neq m}^M p\beta_{m',k}^\text{MK}+N_0}.
\end{align}

To maximize the sum SINR of $\text{UE} \xrightarrow{\text{RIS}} \text{UAV}_{k \in \mathcal K_\text{RIS}}$ links while ensuring the constraints of \eref{RISp}, we assign portions of the RIS elements to the selected UAVs except for the UAV with the most reliability. Let $k^\ast$ be the most reliable UAV, i.e., $k^\ast=\arg\max_{k \in \mathcal K_\text{RIS}}\{\mathcal R_k\}$. Thus, the SINRs of these UAV$_{k\in \mathcal K_\text{RIS}\backslash k^*}$ are just equal to their minimum required SINR. Lastly, the remaining elements of the RIS  can be assigned to the UAV$_{k^\ast}$, providing a higher SINR for that UAV. Therefore, if we rewrite the first constraint in \eref{RISp} 
  using \eqref{appro1}, for all the UAVs except UAV$_{k^\ast}$, as $
\text{SINR}_k(\boldsymbol{\alpha}) = \frac{p[\beta_{m,k}^\text{MK}+\beta_m^\text{MR}\beta^\text{RK}_k \alpha^{2}_{k} N^2z]}{\sum_{m'=1, m'\neq m}^M p\beta_{m',k}^\text{MK}+N_0}= \mathcal R_k\gamma_\text{th}$, we get \eref{closed} with $\Tilde{\gamma}= \sum_{m'=1, m'\neq m}^M p\beta_{m',k}^\text{MK}$. In addition, for UAV$_{k^\ast}$, $\alpha^{*}_{k^*}= 1-\sum_{k=1}^{K_\text{RIS}-1}\alpha^{*}_{k}$, which completes the proof.
\end{proof}


\subsection{UAV Positioning Optimization}
For the given $\mathbf X, \mathbf Z, \boldsymbol{\alpha}^*$, the UAV positioning optimization is given by
\begin{align} \label{UAV}
&\max_{\boldsymbol{u}} ~~~~ \lambda_2(\mathbf L' (\boldsymbol{u}))\\
 & {\rm s.~t.\ }
 ~~~~~\text{SINR}_k\geq \mathcal R_k\gamma_\text{th},~~ \forall k \in \{1, \ldots,  K_\text{RIS}\}. \nonumber 
\end{align} 
As explained in Section II, the Laplacian matrix $\lambda_2(\mathbf L' (\boldsymbol{u}))$ depends on UAV placement which determines the new links. This makes \eref{UAV} a non-convex optimization problem, i.e.,  $\lambda_2(\mathbf L' (\boldsymbol{u}))$ is not a concave function. 
To address the non-concavity of $\lambda_2(\mathbf L' (\boldsymbol{u}))$,  we leverage the graph Laplacian matrix.

Let $\boldsymbol{v} \in \mathbf R^V$
be the eigenvector corresponding to $\lambda_2(\mathbf L )$, which has a unit norm $\|\boldsymbol{v}\|_2=1$ and is orthogonal to the all-ones vector $\mathbf 1$, i.e., $\boldsymbol{1}^T \boldsymbol{v}=0$, where $\boldsymbol{1}$ is $V \times 1$. Since $\mathbf L \boldsymbol{v}= \lambda_2 \boldsymbol{v}$ \cite{new},  $\boldsymbol{v}^T \mathbf L \boldsymbol{v}= \lambda_2 \boldsymbol{v}^T \boldsymbol{v}= \lambda_2$. Hence, $\lambda_2(\mathbf L )$ can be expressed as the smallest eigenvalue that satisfies these conditions \cite{8292633, 4786516}, i.e.,
\begin{align}\label{norm}
\lambda_2(\mathbf L )= \text{Inf}_{\boldsymbol{v}} \big\{\boldsymbol{v}^T \mathbf L \boldsymbol{v},~~ \|\boldsymbol{v}\|_2=1, ~~\boldsymbol{1}^T \boldsymbol{v}=0   \big\},
\end{align}
where Inf is the infimum of a set. \eref{norm} can be
expressed as a convex problem as 
\begin{align}
\label{convex}
&\lambda_2(\mathbf L ) = \min_{\boldsymbol{v} \in \mathbf R^V} ~~~~ \boldsymbol{v}^T \mathbf L \boldsymbol{v}
\\
 & {\rm s.~t.\ }
 ~~~~~\|\boldsymbol{v}\|_2=1, ~ \boldsymbol{1}^T \boldsymbol{v}=0. \nonumber 
\end{align}

Additionally, there is indirect dependence between $\mathbf L' (\boldsymbol{u})$ and $\boldsymbol{u}$. To address this,
we consider that
UAVs are distributed within a $h \times h \times h$ volume. Additionally, the search space along the axes
is uniformly quantized with a step size $\delta$, generating $J$ candidate search
grids for each UAV. This simplifies $\mathbf L' (\boldsymbol{u})$ to be represented by 
\begin{align}\label{L}
 \mathbf L'= \mathbf L + \sum^{K_\text{RIS}}_{k=1} \sum_{j=1}^{J} x_j^{(k)} \mathbf A_j^{(k)} \text{diag}(w_j^{(k)}) \mathbf A_j^{(k)},
\end{align}
where $x_j^{(k)}$ is equal to one if UAV$_k$ is located at the $j$-th grid
point, otherwise $x_j^{(k)}=0$. Moreover, $w_j^{(k)}$ and $\mathbf A_j^{(k)}$ are the
weighting coefficient vectors and the incidence matrix when  UAV$_k$ is deployed at the $j$-th grid point, respectively.

Collecting $x_j^{(k)}$, for $k \in \{1, \ldots, K_\text{RIS}\}$  and $j \in \{1, \ldots, J\}$, in the $JK_\text{RIS} \times 1$ vector $\boldsymbol{x}$, \eref{L} can be written as follows
\begin{align}\label{Lnew}\nonumber 
 \mathbf L' &=  \mathbf L + \sum_{i=1}^{JK_\text{RIS}} x_i \mathbf A_i \text{diag}(w_i) \mathbf A_i \\& 
= \mathbf L + (\boldsymbol{x} \otimes \boldsymbol{1}_V)\Delta,
\end{align}
where 
$\Delta= \big[ (\mathbf A_1 \text{diag}(w_1) \mathbf A_1)^T, (\mathbf A_2 \text{diag}(w_2) \mathbf A_2)^T, \\\ldots, (\mathbf A_{JK_\text{RIS}} \text{diag}(w_{JK_\text{RIS}}) \mathbf A_{JK_\text{RIS}})^T\big]^T$.
Note that since $J$ is much larger than $K_\text{RIS}$, considering the vector $J \times 1$ instead of $JK_\text{RIS} \times 1$ is convenient to provide superior performance. Thus, the problem can be viewed as selecting the optimal $K_\text{RIS}$ UAVs from a set of $J$ candidate UAVs (i.e., $J$ candidate grid points, with each UAV located at the center of its grid). Furthermore, by stacking the SINR levels between the transmitting UE, via the RIS, and UAV$_k$, which is placed at candidate positions in the search grid,  into a $J \times 1$ vector denoted by $\boldsymbol{y}_k$, such that $\text{SINR}_k=\boldsymbol{x}^T \boldsymbol{y}_k$, the optimization problem in \eref{UAV} can be written in terms  of $\boldsymbol{x}$
rather than $\boldsymbol{u}$. Hence, the optimization problem in \eref{UAV} can be written as follows
 \begin{align}\label{UAVnewnew}
&\max_{\boldsymbol{x}} ~~~~ \lambda_2(\mathbf L' (\boldsymbol{x}))\\
 & {\rm s.~t.\ }
 ~~~~~\boldsymbol{x}^T \boldsymbol{y}_k\geq \mathcal R_k \gamma_\text{th},~~ \forall k \in \{1, \ldots,  K_\text{RIS}\}, \nonumber \\ &  ~~~~~~~~~~~\boldsymbol{1}^T\boldsymbol{x}=K_\text{RIS}, ~ \boldsymbol{x} \in \{0,1 \}^{J}. \nonumber 
\end{align}
The combinatorial optimization problem in \eref{UAVnewnew} is intractable due to its high computational complexity. If the constraint
$\boldsymbol{x} \in \{0,1 \}^{J}$ is relaxed to $\boldsymbol{x} \in [0,1]^{J}$, we can
obtain a more tractable and convex optimization problem as follows  
 \begin{align}\label{cnovex1}
&\max_{\boldsymbol{x}} ~~~~ \lambda_2(\mathbf L' (\boldsymbol{x}))\\
 & {\rm s.~t.\ }
 ~~~~~\boldsymbol{x}^T \boldsymbol{y}_k\geq \mathcal R_k \gamma_\text{th},~~ \forall k \in \{1, \ldots,  K_\text{RIS}\}, \nonumber \\ &  ~~~~~~~~~~~\boldsymbol{1}^T\boldsymbol{x}=K_\text{RIS}, ~ 0 \leq \boldsymbol{x} \leq 1. \nonumber 
\end{align}
\begin{proposition}
\eref{cnovex1} is a convex optimization problem and  can be reformulated as an equivalent semi-definite programming (SDP) optimization problem.
 \end{proposition}
\begin{proof}
Using \eref{norm}, the Fiedler value of $\lambda_2(\mathbf L'(\boldsymbol{x}))$   can be expressed
as $\lambda_2(\mathbf L'(\boldsymbol{x}))= \text{Inf}_{\boldsymbol{v}} \big\{\boldsymbol{v}^T \mathbf L'(\boldsymbol{x})) \boldsymbol{v},~~ \|\boldsymbol{v}\|_2=1, ~~\boldsymbol{1}^T \boldsymbol{v}=0   \big\},$ which is the point-wise
infimum of a family of linear functions of $\boldsymbol{x}$ \cite{4786516}. Hence, $\lambda_2(\mathbf L'(\boldsymbol{x}))$ is
a concave function in $\boldsymbol{x}$. Additionally,  the constraints of \eref{cnovex1}
are linear in $\boldsymbol{x}$. Therefore, the optimization problem in \eref{cnovex1} is
a convex optimization problem, and it is equivalent to the following SDP optimization problem \cite{Boyd}
\begin{align}\label{SDP}
&\max_{\mathbf x, s} ~~~~ s
\\
& {\rm s.~t.\ }  ~~~~~ s(\mathbf I - \frac{1}{|\mathbf x|}\mathbf 1 \mathbf 1^T) \preceq \mathbf L'(\mathbf x),\nonumber \\ & 
~~~~~~~~~~~\boldsymbol{x}^T \boldsymbol{y}_k\geq \mathcal R_k \gamma_\text{th},~~ \forall k \in \{1, \ldots,  K_\text{RIS}\}, \nonumber \\& ~~~~~~~~~~~\boldsymbol{1}^T\boldsymbol{x}=K_\text{RIS}, ~ 0 \leq \boldsymbol{x} \leq 1, \nonumber   
\end{align}
where $\mathbf I \in \mathbf R^{V \times V}$ is the identity matrix and $\mathbf F \preceq \mathbf L$ indicates that $\mathbf L- \mathbf F$ is a positive semi-definite matrix and $s$ is an auxiliary variable.
\end{proof}
The optimization problem in \eref{SDP} can be efficiently solved using standard SDP solvers. We first generate all possible mappings between UAVs and their grid points. Second, we use the standard CVX software solver \eref{SDP} to solve SDP optimization \eref{SDP} and obtain $\mathbf x$. After solving, since the entries of output vector $\boldsymbol{x}$ are continuous, we select the  $K_\text{RIS}$   largest values to be  $1$, and the remaining entries are set to $0$. 


		

       
	

\section{Simulation Results}\label{NR}
We consider a 3D scenario where the locations of UAVs are optimized, while the RIS and UE have fixed locations. However, the positions of the RIS and UE are updated in each simulation iteration. In these
simulations, we employ the 3GPP Urban Micro (UMi) model  at a carrier frequency of $3$ GHz to calculate all large-scale
path loss values, consistent with the related works \cite{8989805}. Additionally, we assume a Nakagami$-f$ shape parameter of
$f_1=f_2=5$ for the cascaded UE-RIS-UAV link and $f=1$ for the direct UE-UAV links, with a spread parameter of unity for all the links. Similar to \cite{saifglobecom_E}, we use $\sqrt{\frac{\beta_0}{(d_m^\text{MR})^2}}$ and $\sqrt{\frac{\beta_0}{(d_{k}^\text{RK})^2}}$ for $\text{UE}_m \rightarrow \text{RIS}$ and $\text{RIS} \rightarrow \text{UAV}_{k}$ links, respectively, where $\beta_0$ denotes the path loss at the reference distance $d_\text{ref}=1$ m  and $d$ is the corresponding distance. The RIS  is located at  an altitude of $120$ m and $N=100$ elements. 
Other simulation parameters are given as follows:  $K_\text{RIS}=2$,
		$P = 30$ dBm, $p=23$ dBm, $\beta_0=10^{-2}$,   the channel bandwidth is 250 KHz,  $J=40$, $\gamma_\text{th}=60$ dB, and  $N_0=-120$ dBm.
    
The performance of the proposed scheme is compared with  (i) the scheme proposed in \cite{saifglobecom_E} that utilizes the whole RIS to create one link only; (ii)) the original scheme without the RIS, inspired by \cite{8292633}; and (iii) the proposed scheme with random UAV locations.

 







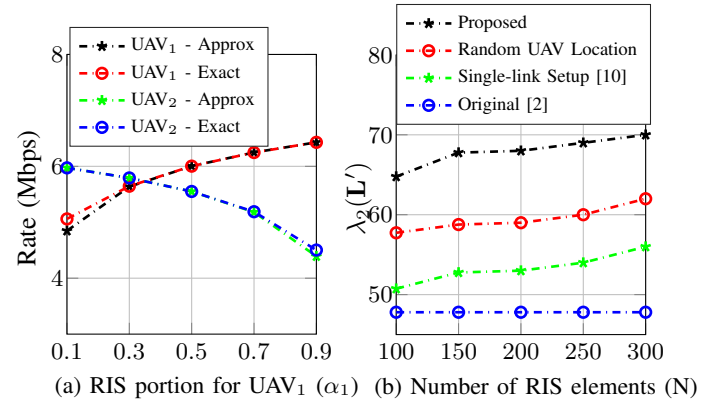
\begin{figure}[t!]
    \begin{minipage}[b]{0.48\linewidth}
  \centering 
    \centerline{
%
%
\begin{tikzpicture}

\begin{axis}[%
width=33mm,
height=37mm,
at={(0mm,0mm)},
scale only axis,
xmin=0.1,
xmax=0.9,
xtick={0.1,0.3,  0.5,  0.7,   0.9},
ymin=3,
ymax=8,
yminorticks=true,
yticklabel style = {font=\small,xshift=0.5ex},
xticklabel style={
    /pgf/number format/fixed,
    /pgf/number format/precision=1,
    font=\small,
    yshift=0ex
},
scaled ticks=false,
axis background/.style={fill=white},
xmajorgrids,
ymajorgrids,
legend style={font=\scriptsize, at={(0.01,0.67)}, anchor=south west,legend cell align=left, align=left, draw=white!15!black}
]

\addplot [color=black, dashdotted, line width=1pt, mark=star, mark options={solid, black}]
table[row sep=crcr]{%
0.1    4.8473\\
0.3     5.6358\\
0.5         6.0035\\
0.7           6.2478 \\
0.9             6.4288  \\
};
\addlegendentry{UAV$_1$ - Approx}

\addplot [color=red, dashdotted, line width=1pt, mark=o, mark options={solid, red}]
  table[row sep=crcr]{%
0.1    5.0596\\
0.3     5.6446\\
0.5         6.0030\\
0.7           6.2468 \\
0.9            6.4287 \\
};
\addlegendentry{UAV$_1$ - Exact}

\addplot [color=green, dashdotted, line width=1pt, mark=star, mark options={solid, green}]
  table[row sep=crcr]{%
0.1    5.9729\\
0.3      5.7922\\
0.5         5.5503\\
0.7           5.1807\\
0.9             4.3917\\
};
\addlegendentry{UAV$_2$ - Approx}

\addplot [color=blue, dashdotted, line width=1pt, mark=o, mark options={solid, blue}]
  table[row sep=crcr]{%
0.1    5.9726\\
0.3     5.7921\\
0.5        5.5504\\
0.7           5.1883 \\
0.9            4.5026 \\
};
\addlegendentry{UAV$_2$ - Exact}

\end{axis}
\node[rotate=0,fill=white] (BOC6) at (1.850cm,-.7cm){\small (a) RIS portion for UAV$_1$ ($\alpha_1$)};
\node[rotate=90] at (-5mm,18mm){Rate (Mbps)};

\end{tikzpicture}%
}
    \centerline{}
    \end{minipage}
    \begin{minipage}[b]{0.48\linewidth}
  \centering 
    \centerline{
%
%
\begin{tikzpicture}

\begin{axis}[%
width=33mm,
height=37mm,
at={(0mm,0mm)},
scale only axis,
xmin=100,
xmax=300,
xtick={100, 150, 200, 250, 300},
ymin=45,
ymax=80,
yminorticks=true,
yticklabel style = {font=\small,xshift=0.5ex},
xticklabel style = {font=\small,yshift=0ex},
axis background/.style={fill=white},
xmajorgrids,
ymajorgrids,
legend style={font=\scriptsize, at={(0.0,0.75)}, anchor=south west,legend cell align=left, align=left, draw=white!15!black}
]

\addplot [color=black, dashdotted, line width=1pt, mark=star, mark options={solid, black}]
  table[row sep=crcr]{%
100             64.7696\\
150              67.7696\\
200              68. 1234\\
250              69. 1234\\
300              70. 1234\\
};
\addlegendentry{Proposed}

\addplot [color=red, dashdotted, line width=1pt, mark=o, mark options={solid, red}]
  table[row sep=crcr]{%
100             57.7396\\
150              58.7696\\
200              59. 1234\\
250              60. 1234\\
300              62. 1234\\
};
\addlegendentry{Random UAV Location}

\addplot [color=green, dashdotted, line width=1pt, mark=star, mark options={solid, green}]
  table[row sep=crcr]{%
100             50.7380\\
150              52.7696\\
200              53. 1234\\
250              54. 1234\\
300              56. 1234\\
};
\addlegendentry{Single-link Setup \cite{saifglobecom_E}}

  partition                                                           
\addplot [color=blue, dashdotted, line width=1pt, mark=o, mark options={solid, blue}]
  table[row sep=crcr]{%
100             47.7970\\
150              47.7970\\
200              47.7970\\
250              47.7970\\
300              47.7970\\
};
\addlegendentry{Original \cite{8292633}}

\end{axis}
\node[rotate=0,fill=white] (BOC6) at (1.850cm,-.7cm){\small (b) Number of RIS elements (N)};
\node[rotate=90] at (-5mm,18mm){$\lambda_2(\mathbf L')$};
\end{tikzpicture}%
}
    \centerline{} 
    \end{minipage}
    \vspace{-5mm}
    \caption{ (a) Rate  versus $\alpha_{1}$ and (b) network connectivity versus $N$ for a network with $12$ UAVs, $\gamma_\text{th}=60$ dB, $\gamma^\text{UAV}_{0}=75$ dB, and $\gamma^\text{UE}_{0}=70$ dB.}
    \label{fig3} \vspace{-4mm}
\end{figure}

We present some numerical results to assess the
approximation of removing the non-aligned term  in \eqref{exact}. We plot the exact rate in \eqref{exact}, while including the non-aligned term, and the approximated rate using  \eref{appro1}. For the purpose of
this part, the fixed 3D coordinates in meters for the UE are $(118, 220, 0)$, while  UAV$_1$, UAV$_2$, and the RIS are located at $(160, 140, 200)$, $(170, 14, 200)$, and $(0, 0, 120)$, respectively. In Fig. \ref{fig3}(a), the  presented results are averaged over $10^5$ simulations. For UAV$_1$ and UAV$_2$, the RIS partitions, respectively, are $\alpha_{1}$ and  $\alpha_{2}=1-\alpha_{1}$.  We adopt the  phase shift quantization levels similar to \cite{10507188} with $4$ bits. From the figure, we notice that approximated rates for the UAVs tightly match the exact rates for most values of $\alpha_{1}$. Overall, these results justify our assumption to (i) ignore the impact of non-aligned channels from the other portions of the RIS and (ii) approximate \eref{exact} by \eref{appro1}. The figure shows that as $\alpha_1$ increases, more RIS elements are allocated to support the link to UAV$_1$. As a result, the exact and the approximated rates of  $\text{UE}\xrightarrow{\text{RIS}} \text{UAV}_{1}$ link increase, while both rates of  $\text{UE}\xrightarrow{\text{RIS}} \text{UAV}_{2}$ link decrease.

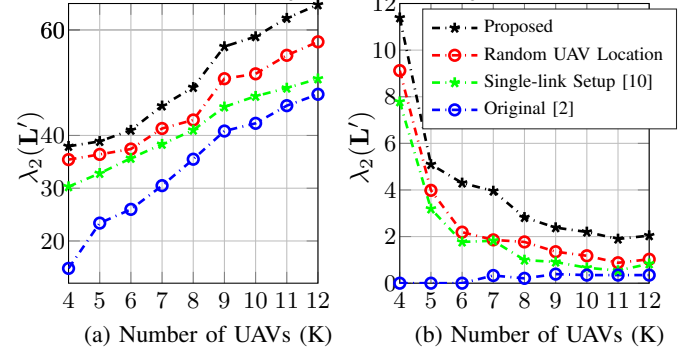
\begin{figure}[t!]
    \begin{minipage}[b]{0.48\linewidth}
  \centering 
    \centerline{
%
%
\begin{tikzpicture}

\begin{axis}[%
width=33mm,
height=37mm,
at={(0mm,0mm)},
scale only axis,
xmin=4,
xmax=12,
xtick={4, 5, 6, 7, 8, 9, 10, 11, 12},
ymin=12,
ymax=65,
ytick={10,20,30,40,60, 60, 70},
yminorticks=true,
yticklabel style = {font=\small,xshift=0.5ex},
xticklabel style = {font=\small,yshift=0ex},
axis background/.style={fill=white},
xmajorgrids,
ymajorgrids,
legend style={font=\scriptsize, at={ (0.0,0.95) }, anchor=south west,legend cell align=left, align=left, draw=white!15!black}
]

\addplot [color=black, dashdotted, line width=1pt, mark=star, mark options={solid, black}]
  table[row sep=crcr]{%
4    37.9458\\
5    38.8763\\
6     41.0019\\
7       45.6034\\
8         49.1102\\
9           56.8432\\
10           58.7338 \\
11           62.2174\\
12             64.7696\\
};

\addplot [color=red, dashdotted, line width=1pt, mark=o, mark options={solid, red}]
  table[row sep=crcr]{%
4    35.4221\\
5    36.4081\\
6     37.4218\\
7       41.3178\\
8         42.9317 \\
9           50.7529\\
10           51.6750 \\
11           55.1846\\
12             57.7396\\
};

\addplot [color=green, dashdotted, line width=1pt, mark=star, mark options={solid, green}]
  table[row sep=crcr]{%
4    30.2945\\
5    32.8197\\
6     35.6775\\
7       38.3364\\
8         41.0114\\
9           45.4151\\
10           47.4308\\
11          48.9850\\
12             50.7380\\
};

\addplot [color=blue, dashdotted, line width=1pt, mark=o, mark options={solid, blue}]
  table[row sep=crcr]{%
4    14.8058\\
5    23.3869\\
6     25.9932\\
7       30.4880\\
8         35.4937\\
9           40.8344\\
10           42.2960\\
11           45.6418\\
12             47.7970\\
};

\end{axis}
\node[rotate=0,fill=white] (BOC6) at (1.850cm,-.7cm){\small (a) Number of UAVs (K)};
\node[rotate=90] at (-5mm,18mm){$\lambda_2(\mathbf L')$};
\end{tikzpicture}%
}
    \centerline{}
    \end{minipage}
    \begin{minipage}[b]{0.48\linewidth}
  \centering 
    \centerline{
%
%
\begin{tikzpicture}

\begin{axis}[%
width=33mm,
height=37mm,
at={(0mm,0mm)},
scale only axis,
xmin=4,
xmax=12,
xtick={4, 5, 6, 7, 8, 9, 10, 11, 12},
ymin=0,
ymax=12,
ytick={0,2,4,6,8, 10, 12},
yminorticks=true,
yticklabel style = {font=\small,xshift=0.5ex},
xticklabel style = {font=\small,yshift=0ex},
axis background/.style={fill=white},
xmajorgrids,
ymajorgrids,
legend style={font=\scriptsize, at={(0.09,0.55)}, anchor=south west,legend cell align=left, align=left, draw=white!15!black}
]

\addplot [color=black, dashdotted, line width=1pt, mark=star, mark options={solid, black}]
  table[row sep=crcr]{%
4    11.3882\\
5     5.0947\\
6     4.3012\\
7       3.9545\\
8         2.8228\\
9           2.3842\\
10           2.1985 \\
11           1.8975\\
12             2.0415\\
};
\addlegendentry{Proposed}

\addplot [color=red, dashdotted, line width=1pt, mark=o, mark options={solid, red}]
  table[row sep=crcr]{%
4    9.1216\\
5    3.9839\\
6     2.1898\\
7       1.8583\\
8         1.7705\\
9           1.3504\\
10           1.1763 \\
11           0.8766\\
12             1.0255\\
};
\addlegendentry{Random UAV Location}

\addplot [color=green, dashdotted, line width=1pt, mark=star, mark options={solid, green}]
  table[row sep=crcr]{%
4    7.7846\\
5     3.1899\\
6      1.7714\\
7       1.8211\\
8         0.9953\\
9           0.9094\\
10           0.6629\\
11          0.5416\\
12             0.8322\\
};
\addlegendentry{Single-link Setup \cite{saifglobecom_E}}

\addplot [color=blue, dashdotted, line width=1pt, mark=o, mark options={solid, blue}]
  table[row sep=crcr]{%
4    0.0000\\
5    0.0000\\
6     0.0000\\
7       0.3297\\
8         0.2079\\
9           0.3897\\
10           0.3479\\
11           0.3479\\
12             0.3479\\
};
\addlegendentry{Original \cite{8292633}}

\end{axis}
\node[rotate=0,fill=white] (BOC6) at (1.850cm,-.7cm){\small (b) Number of UAVs (K)};
\node[rotate=90] at (-5mm,18mm){$\lambda_2(\mathbf L')$};
\end{tikzpicture}%
}
    \centerline{} 
    \end{minipage}
    \vspace{-5mm}
    \caption{Network connectivity versus $K$ for (a) well-connected network with   $\gamma^\text{UAV}_{0}=75$ dB  and $\gamma^\text{UE}_{0}=70$ dB and (b) sparse network with    $\gamma^\text{UAV}_{0}=\gamma^\text{UE}_{0}=83$ dB.}
    \label{fig4} \vspace{-4mm}
\end{figure}

Fig. \ref{fig3}(b) shows the network connectivity versus the number of elements $N$ for $K=12$. We can see that the increase in the number of RIS elements improves connectivity, since  RISs with more elements can boost up the quality of the newly established $\text{UE} \rightarrow \text{UAV}_{k}$  links. As a result, the network connectivity is increased. Original scheme is fixed.

Fig. \ref{fig4} shows the network connectivity versus the number of UAVs $K$ for two scenarios: (a) well-connected network and (b) sparse network, with $40$ grid points for the UAVs $K_\text{RIS}$. In the well-connected scenario, we assume high connectivity among the UAVs and between the UE and the UAVs, meaning direct links between the UE and the UAV(s) are always present.  In this case, the original network forms one connected graph. As shown, the proposed scheme consistently outperforms other schemes due to the joint optimization, which enables (i) the RIS to establish two potential links and (ii) optimal UAV positioning that improves the quality of these links. In contrast, the single-link setup shows noticeable degradation compared to the proposed scheme with random UAV placement.

For the sparse network, we set $\gamma^\text{UAV}_{0}=\gamma^\text{UE}_{0}=83$ dB. In this case, the original network is not always fully connected, occasionally resulting in zero connectivity. Therefore, the performance of the original scheme fluctuates between zero and positive values, but the presented results are averaged over a large number of Monte Carlo simulations. Since the network is sparse due to the high SINR thresholds for  the $\text{UE} \rightarrow \text{UAV}$ and $\text{UAV}_k \rightarrow \text{UAV}_{k'}$ links, adding more UAVs does not lead to increased connectivity as observed in Fig.~\ref{fig4}(b).   This is because deploying more UAVs, while many of their connections fail to meet the high SINR thresholds, results in a sparser network graph (i.e., more nodes but fewer edges). Consequently, the network connectivity of the RIS-assisted schemes decreases, in contrast to Fig.~\ref{fig4}(a), where the network graph is highly connected. 

Fig.~\ref{fig8} illustrates the impact of imperfect CSI coefficients $\sigma^2_e$ on network connectivity for $K=10$. We consider the imperfect CSI model in \cite{10507188}.  We assume $\sigma^2_{e, \text{MK}}=\sigma^2_{e, \text{MRK}}=\sigma^2_{e}$, and present the network connectivity for both perfect CSI ($\sigma^2_{e}=0$) and imperfect CSI as a function of $\sigma^2_{e}$. As shown in the figure, imperfect CSI causes a gradual performance degradation compared to the perfect CSI case. However, the impact remains relatively limited until $\sigma_e^2 > 10^0$, beyond which a more noticeable degradation is observed. This is because CSI errors primarily affect the SINRs of the RIS-assisted links, while the overall network connectivity also depends on the underlying graph topology, making it relatively robust to moderate CSI imperfections.

\begin{figure}[!t]
  \centering  
\centerline{
%
%
\begin{tikzpicture}

\begin{axis}[%
width=65mm,
height=45mm,
at={(0mm,0mm)},
scale only axis,
xmode=log,
xmin=0.0001,
xmax=1,
ymin=37,
ymax=40,
yminorticks=true,
yticklabel style = {font=\small,xshift=0.5ex},
xticklabel style = {font=\small,yshift=0ex},
axis background/.style={fill=white},
xmajorgrids,
ymajorgrids,
legend style={font=\scriptsize, at={(0.0,0.35)}, anchor=south west,legend cell align=left, align=left, draw=white!15!black}
]

\addplot [color=black, dashdotted, line width=1pt, mark=star, mark options={solid, black}]
  table[row sep=crcr]{%
 0.0001     39.5243    \\
0.001        39.5243   \\
0.01       39.5243    \\
0.1       39.5243    \\
1       39.5243    \\
};
\addlegendentry{Proposed - Perfect CSI}

\addplot [color=red, dashdotted, line width=1pt, mark=star, mark options={solid, red}]
  table[row sep=crcr]{%
0.0001     39.5243    \\
0.001        39.4243   \\
0.01       39.3243    \\
0.1       39.1243    \\
1       39.0243    \\
};
\addlegendentry{Proposed - Imperfect CSI}

\end{axis}
\node[rotate=0,fill=white] (BOC6) at (3.35cm,-.78cm){\small   $\sigma^2_e$};
\node[rotate=90] at (-8mm, 22mm){$\lambda_2(\mathbf L')$};
\end{tikzpicture}%
}
\caption{Network connectivity versus $\sigma^2_e$ for $K=10$, $\gamma^\text{UAV}_{0}=\gamma^\text{UE}_{0}=83$ dB.}
\label{fig8}
\end{figure}
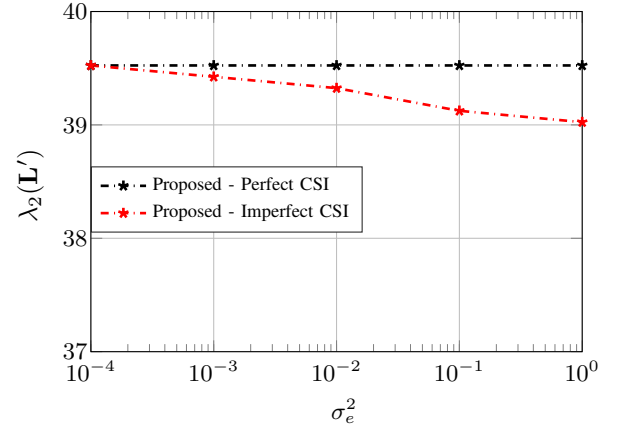
 

\section{conclusion}\label{C}
This work  formulates and solves the network connectivity maximization problem,  characterized by the Fiedler value, through a joint optimization of RIS-aided link selection, RIS  partitioning, and UAV positioning. With the assistance of RIS deployment and UAV positioning, we aim to enhance the network connectivity of uplink RIS-assisted UAV networks. An iterative solution is developed and benchmarked against several schemes, including the original method, random UAV placements, and a single-link setup. Concerning future work, the employment of scheduling multiple transmitting UEs to a single RIS, rather than assigning one transmitting UE per RIS, may be taken into account.

\bibliographystyle{IEEEtran}
\bibliography{main}

@INPROCEEDINGS{10437618,
  author={Javad-Kalbasi, Mohammad and Al-Abiad, Mohammed S. and Valaee, Shahrokh},
  booktitle={GLOBECOM 2023 - 2023 IEEE Global Communications Conference}, 
  title={Energy Efficient Communications in {RIS}-Assisted {UAV} Networks Based on Genetic Algorithm}, 
  year={2023},
  volume={},
  number={},
  pages={5901-5906},
  doi={10.1109/GLOBECOM54140.2023.10437618}}

@ARTICLE{9756313,
  author={Obeed, Mohanad and Chaaban, Anas},
  journal={IEEE Transactions on Wireless Communications}, 
  title={Joint Beamforming Design for Multiuser {MISO} Downlink Aided by a Reconfigurable Intelligent Surface and a Relay}, 
  year={2022},
  volume={21},
  number={10},
  pages={8216-8229},
  doi={10.1109/TWC.2022.3164903}}

@INPROCEEDINGS{10104574,
  author={Weinberger, Kevin and Reifert, Robert-Jeron and Sezgin, Aydin and Basar, Ertugrul},
  booktitle={WSA \& SCC 2023; 26th International ITG Workshop on Smart Antennas and 13th Conference on Systems, Communications, and Coding}, 
  title={{RIS}-enhanced Resilience in Cell-Free {MIMO}}, 
  year={2023},
  volume={},
  number={},
  pages={1-6},
  doi={}}

@ARTICLE{9716042,
  author={Hassan, Md. Zoheb and Kaddoum, Georges and Akhrif, Ouassima},
  journal={IEEE Internet of Things Journal}, 
  title={Interference Management in Cellular-Connected Internet of Drones Networks With Drone-Pairing and Uplink Rate-Splitting Multiple Access}, 
  year={2022},
  volume={9},
  number={17},
  pages={16060-16079},
  doi={10.1109/JIOT.2022.3152382}}

@INPROCEEDINGS{8292633,
  author={Abdel-Malek, Mai A. and Ibrahim, Ahmed S. and Mokhtar, Mohamed},
  booktitle={2017 IEEE 28th Annual International Symposium on Personal, Indoor, and Mobile Radio Communications (PIMRC)}, 
  title={Optimum {UAV} positioning for better coverage-connectivity tradeoff}, 
  year={2017},
  volume={},
  number={},
  pages={1-5},
  doi={10.1109/PIMRC.2017.8292633}}

@ARTICLE{4786516,
  author={Ibrahim, Ahmed S. and Seddik, Karim G. and Liu, K.J. Ray},
  journal={IEEE Transactions on Wireless Communications}, 
  title={Connectivity-aware network maintenance and repair via relays deployment}, 
  year={2009},
  volume={8},
  number={1},
  pages={356-366},
  doi={10.1109/T-WC.2009.080045}}

@ARTICLE{4657335,
  author={Pandana, Charles and Liu, K.J. Ray},
  journal={IEEE Transactions on Wireless Communications}, 
  title={Robust connectivity-aware energy-efficient routing for wireless sensor networks}, 
  year={2008},
  volume={7},
  number={10},
  pages={3904-3916},
  doi={10.1109/T-WC.2008.070453}}

@INPROCEEDINGS{new,
  author={M. Fiedler},
  booktitle={Czechoslovak Mathematical J.}, 
  title={Algebraic connectivity of graphs}, 
  year={1973},
  volume={23},
  number={},
  pages={298-305}}

@ARTICLE{10458024,
  author={Naim Shaikh, Mohd Hamza and Celik, Abdulkadir and Eltawil, Ahmed M. and Nauryzbayev, Galymzhan},
  journal={IEEE Transactions on Wireless Communications}, 
  title={Grant-Free {NOMA} Through Optimal Partitioning and Cluster Assignment in STAR-RIS Networks}, 
  year={2024},
  volume={23},
  number={8},
  pages={10166-10181},
  doi={10.1109/TWC.2024.3369186}}

@ARTICLE{saifglobecom_E,
  author={Saif, Mohammed and Javad-Kalbasi, Mohammad and Valaee, Shahrokh},
  journal={IEEE Transactions on Wireless Communications}, 
  title={Effectiveness of Reconfigurable Intelligent Surfaces to Enhance Connectivity in {UAV} Networks}, 
  year={2024},
  volume={23},
  number={12},
  pages={18757-18773},
  doi={10.1109/TWC.2024.3476422}}

@ARTICLE{10507188,
  author={Arzykulov, Sultangali and Celik, Abdulkadir and Nauryzbayev, Galymzhan and Eltawil, Ahmed M.},
  journal={IEEE Transactions on Cognitive Communications and Networking}, 
  title={Aerial {RIS}-Aided Physical Layer Security: {O}ptimal Deployment and Partitioning}, 
  year={2024},
  volume={},
  number={},
  pages={1-1},
  doi={10.1109/TCCN.2024.3392798}}

@INPROCEEDINGS{Boyd,
  author={S. Boyd},
  booktitle={Proc.
International Congress of Mathematicians,}, 
  title={Convex optimization of graph laplacian eigenvalues}, 
  year={2006},
  volume={3},
  number={},
  pages={1311-1319}}

@ARTICLE{8989805,
  author={Lyu, Jiangbin and Zhang, Rui},
  journal={IEEE Wireless Communications Letters}, 
  title={Spatial Throughput Characterization for Intelligent Reflecting Surface Aided Multiuser System}, 
  year={2020},
  volume={9},
  number={6},
  pages={834-838},
  doi={10.1109/LWC.2020.2972527}}

\end{document}